\documentclass[
final
]{dmtcs-episciences}

\usepackage[utf8]{inputenc}
\usepackage{subfigure}
\usepackage[ruled,vlined,noend]{algorithm2e}

\usepackage[round]{natbib}

\newtheorem{theorem}{Theorem}[section]

\newtheorem{lemma}[theorem]{Lemma}
\newtheorem{prop}[theorem]{Property}

\newtheorem{obs}[theorem]{Observation}
\newtheorem{corollary}[theorem]{Corollary}

\author{Paul Dorbec\affiliationmark{1}
  \and Antonio Gonz\'alez\affiliationmark{2}
  \and Claire Pennarun\affiliationmark{3}
  }
\title[Power domination in maximal planar graphs]{Power domination in maximal planar graphs}
\affiliation{
  Normandie Univ, UNICAEN, ENSICAEN, CNRS, GREYC, France\\
  Departamento de Didáctica de las Matemáticas, Universidad de Sevilla, Spain\\
  LaBRI, Univ. Bordeaux, France}
\received{2019-1-27}
\revised{2019-10-25}
\accepted{2019-10-27}
\begin{document}
\publicationdetails{21}{2019}{4}{18}{5127}
\maketitle
\begin{abstract}

   Power domination in graphs emerged from the problem of monitoring an electrical system by placing as few measurement devices in the system as possible.
    It corresponds to a variant of domination that includes the possibility of propagation.
    For measurement devices placed on a set $S$ of vertices of a graph $G$, the set of monitored vertices is initially the set $S$ together with all its neighbors.
    Then iteratively, whenever some monitored vertex $v$ has a single neighbor $u$ not yet monitored, $u$ gets monitored.
    A set $S$ is said to be a power dominating set of the graph $G$ if all vertices of $G$ eventually are monitored.
    The power domination number of a graph is the minimum size of a power dominating set.
    In this paper, we prove that any maximal planar graph of order $n \geq 6$ admits a power dominating set of size at most $\frac{n-2}{4}$.
\end{abstract}

\section{Introduction}

The notion of \emph{power domination} arose in the context of monitoring an electrical network~(\cite{baldwin_93, mili_90, phadke_86}), i.e., knowing the state of each component (e.g. the voltage magnitude at loads) by measuring some variables such as currents and voltages. The measurements are done by placing Phasor Measurement Units (PMUs) at selected locations.
PMUs monitor the state of the adjacent components, then with the use of electrical laws (such as Ohm's and Kirschoff's Laws), it is possible to determine the state of components further away in the network.
Since PMUs are costly, it is important to monitor a graph with as few PMU as possible. In this paper, we consider the problem of monitoring maximal planar graphs with few PMUs. Before getting further into technical details, we need the following graph definitions.
\smallskip

Let $G=(V(G),E(G))$ be a finite, simple, and undirected graph of order $n = |V (G)|$.
The \emph{open neighborhood} of a vertex $u\in V(G)$ is $N_G(u) = \{v \in V(G) \mid uv \in E(G)\}$, and its \emph{closed neighborhood} is $N_G[u] = N_G(u) \cup \{u\}$; the \emph{open} and \emph{closed neighborhood of} a subset of vertices $S$ is $N_G(S)=\bigcup_{v \in S} N_G(v)$
and $N_G[S]=S\cup N_G(S)$, respectively.
The subgraph of $G$ induced by $S$ is written $G[S]$ (the subscript $G$ is dropped from the notations when no confusion may arise).

The electrical network monitoring problem was transposed into graph-theore-tical terms by \cite{haynes_02}.
Originally, the definition of power domination ensured the monitoring of the edges as well as of the vertices, and contained many propagation rules.
Here, we consider an equivalent definition from~\cite{brueni_05} that only requires monitoring the vertices.
Given a graph $G$ and a set $S\subseteq V(G)$, we build a set $M_G(S)$ (or simply $M(S)$ when the graph $G$ is clear from the context) as follows:
at first, $M_G(S) = N[S]$, and then iteratively a vertex $u$ is added to $M_G(S)$ if $u$ has a neighbor $v$ in $M_G(S)$ such that $u$ is the only neighbor of $v$ not in $M_G(S)$ (we say that $v$ \emph{propagates} to $u$).
At the end of the process, we say that $M_G(S)$ is the set of vertices \emph{monitored} by $S$; the \emph{non-monitored} vertices are those of the set $V(G) \setminus M_G(S)$.
We say that $G$ is \emph{monitored} by $S$ when $M_G(S)=V(G)$ and, in that case, $S$ is said to be a \emph{power dominating set} of $G$. The minimum cardinality of such a set is the \emph{power domination number} of $G$, denoted by $\gamma_{\rm P}(G)$.

\medskip
The decision problem \textsc{Power Dominating Set} naturally associated to power domination (i.e., ``Given a graph $G$ and an integer $k$, does $G$ have a power dominating set of order at most $k$?'') was proven NP-complete, by a reduction from the 3-SAT problem~(\cite{haynes_02,liao_05}) (giving NP-completeness of the problem on bipartite graphs, chordal graphs and split graphs). A reduction from \textsc{Dominating Set} was also given~(\cite{guo_08, kneis_06}), that implies the NP-completeness when restricted to planar graphs or circle graphs.
However, polynomial algorithms were proposed to compute the power domination number of trees~(\cite{haynes_02,guo_08}), block graphs~(\cite{xu_06}), interval graphs~(\cite{liao_05}), and circular-arc graphs~(\cite{liao_05,liao_13}).

Concerning the parameter $\gamma_{\rm P}(G)$, tight upper bounds are also known for particular classes:
$\gamma_{\rm P}(G) \leq \frac{n}{3}$ if $G$ is connected (\cite{zhao_06}) or a tree (\cite{haynes_02}),
whereas cubic graphs satisfy $\gamma_{\rm P}(G) \leq \frac{n}{4}$~(\cite{dorbec_13}).
Furthermore, the exact value of $\gamma_{\rm P}(G)$ has been determined for regular grids and their generalizations:
square grid~(\cite{dorfling_06}) and other products of paths~(\cite{dorbec_08}), hexagonal grids~(\cite{ferrero_11}), as well as cylinders and tori~(\cite{barrera_11}).
The only known results for general planar graphs concern graphs with diameter two or three~(\cite{zhao_07}).

\medskip

A graph $G$ is a \emph{planar graph} if it admits a crossing-free embedding in
the plane.
When the addition to $G$ of any edge would result in a non-planar graph, $G$ is said to be a \emph{maximal planar graph}.
A planar graph $G$ together with a crossing-free embedding on the plane is called a \emph{plane graph}, or a \emph{triangulation}
when $G$ is a maximal planar graph.
For any subset $S\subseteq V(G)$, the graph $G[S]$
can be viewed as a plane graph with the embedding inherited from the embedding of $G$.
The only unbounded face is called the \emph{outer face} of $G$,
and the vertices of $G$ are called respectively \emph{exterior} or \emph{interior} depending on whether
they belong to the outer face or not.
We say that a subgraph of a triangulation $G$ is \emph{facial} if all of its faces but the outer face are also faces of $G$. In particular, we denote by $[uvw]$ a facial triangle formed by vertices $u$, $v$ and $w$ in $G$.
Note that power dominating sets are independent of the embedding of the graph as they only depend on vertex adjacencies. We only make use of the embedding of the graph in the proofs.

\medskip
The main result of this paper consists in the following theorem:

\begin{theorem} \label{th:main}
    If $G$ is a maximal planar graph of order $n\geq 6$, then $\gamma_{\rm P}(G) \leq \frac{n-2}{4}$.
\end{theorem}

The bound of Theorem~\ref{th:main} is tight for graphs on six vertices.
We also know of one graph on ten vertices for which this bound is tight, the \emph{triakis tetrahedron} drawn in Figure~\ref{fig:triakis_seul}.
\begin{figure}
\centering
\includegraphics[scale=0.7,page=2]{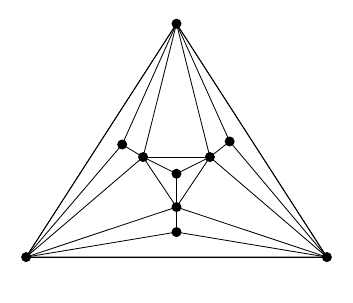}
\caption{The triakis tetrahedron, having ten vertices and power domination number two.}
\label{fig:triakis_seul}
\end{figure}
To propose a general family of maximal planar graphs that have large power domination number,
we use the configurations of Figure~\ref{fig:bad-guys}.
Observe that if one of these configurations $H$ is a facial subgraph of $G$,
then any power dominating set of $G$ contains one of the vertices of $H$.
Otherwise, even if all the exterior vertices are monitored,
they can not propagate to any of the interior vertices of $H$.
Thus $\gamma_{\rm P}(G)$ is at least the number of disjoint facial special configurations in $G$.
Taking many disjoint copies of the two first configurations (that have six vertices)
and then completing the graph into a triangulation
by arbitrarily adding edges between external vertices of the configurations (see Figure~\ref{fig:tri_octa}),
we obtain a family of graphs that have power domination number $\frac{n}{6}$.
Note that this construction is similar to the construction for classical domination given in \cite{matheson_96} reaching the bound $\gamma(G)=\frac{n}{4}$.
As a consequence, and thanks to Theorem~\ref{th:main}, we also get the following result:

\begin{theorem}
    For $n \geq 6$, every maximal planar graph with $n$ vertices has a power dominating set containing at most $\alpha(n)$ vertices, with $\frac{n}{6} \leq \alpha(n) \leq \frac{n-2}{4}$.
\end{theorem}
Determining the best possible value of $\alpha(n)$ remains an open problem.

\begin{figure}[ht]
    \centering
    \includegraphics[scale=0.7]{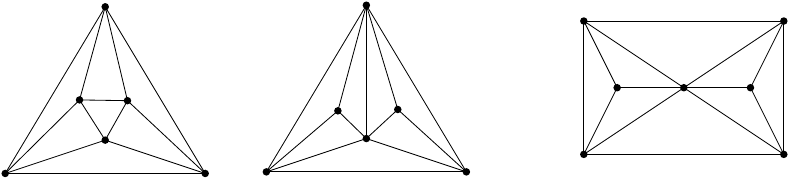}
    \caption{The good, the bad and the ugly configurations in a triangulation.}
    \label{fig:bad-guys}
\end{figure}

\begin{figure}[ht]
    \centering
    \includegraphics[scale=0.9]{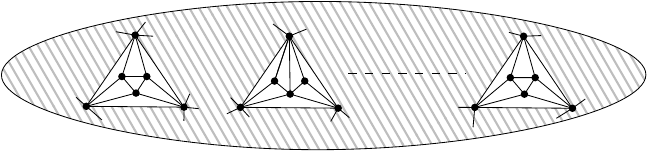}
    \caption{A class of maximal planar graphs for which $\gamma_{\rm P}(G)=\frac{n}{6}$. The hatched area is triangulated arbitrarily.}
    \label{fig:tri_octa}
\end{figure}

The proof of Theorem~\ref{th:main} is done in three distinct steps, each of them described in a separate algorithm in Section~\ref{sec:algo}.
 The first algorithm deals with the special configurations formed by overlapping configurations from Figure~\ref{fig:bad-guys}.
These special configurations are characterized in Section~\ref{sec:b-vertices}.
The end of the proof relies on a final Lemma that is proved in Section~\ref{sec:proof}.

\section{Identifying bad guys}\label{sec:b-vertices}

Our algorithm deals first with some special configurations, that are the possible intersections of the configurations from Figure~\ref{fig:bad-guys}. We here characterize these special configurations.
Note that a facial octahedron (third configuration of Figure~\ref{fig:bad-guys}) may only share vertices of its outer face with other configurations.
We thus focus on the other two configurations.

We call \emph{3-vertex} a vertex of $G$ with degree 3, and therefore whose neighborhood induces a $K_4$.
A \emph{b-vertex} is any vertex $u \in V(G)$ with exactly two 3-neighbors $v$ and $v'$,
and such that $N[u] = N[v] \cup N[v']$.
Note that b-vertices have degree at most six and
their neighborhood necessarily induces one of the subgraphs of Figure~\ref{fig:b-vertex}.
In all figures of this section, b-vertices are depicted with blue squares, and 3-vertices are drawn white.

\begin{figure}[h]
\centering
\includegraphics[scale=0.8]{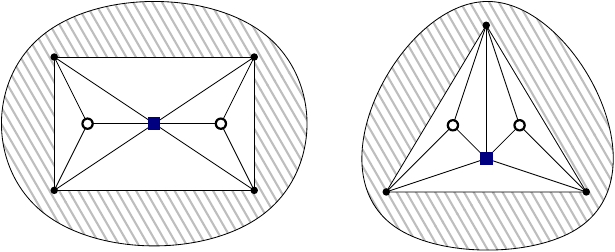}
\caption{The two possible neighborhoods of a b-vertex $v$.} 
\label{fig:b-vertex}
\end{figure}

\begin{obs} \label{obs:adj}
Any two b-vertices $u,u' \in V(G)$ are adjacent if and only if there exists a 3-vertex $v \in N(u) \cap N(u')$.
\end{obs}

\begin{proof}
  By definition of b-vertices, if $u'$ is adjacent to $u$, it is also adjacent to a 3-vertex $v$ adjacent to $u$, and so $u$ and $u'$ have a common neighbor of degree 3.
  Moreover, if $v$ has degree 3, all its of $v$ are pairwise adjacent, and thus two b-vertices $u$ and $u'$ that have $v$ as a common neighbor are adjacent.
\end{proof}

\begin{lemma}\label{lem:shared-3-neighbors}
  If $G$ contains two 3-vertices $v_1,v_2$ with two common b-neighbors, then $G$ is isomorphic to one of the graphs depicted in Figure~\ref{fig:b-vertex_common}.
\end{lemma}

\begin{proof}
  Either $v_1$ and $v_2$ have three common neighbors (inducing the first subgraph), or they have distinct third neighbors (inducing the second subgraph). In the first subgraph, all triangles are incident to a 3-vertex, so they are facial and there is no possibility for more vertices in the graph. In the second subgraph, the only faces not incident to a 3-vertex are incident to  a b-vertex, which can not have other neighbors. Again, all triangles must then be facial.
\end{proof}

\begin{figure}[ht]
  \centering
  \includegraphics[scale=0.8]{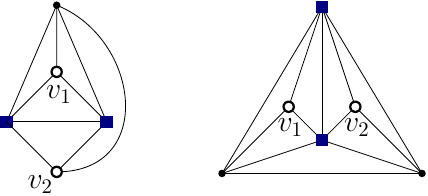}
  \caption{The two
  possibilities for $G$
  if two 3-vertices $v_1$ and $v_2$ have two common b-neighbors. In both cases, $\gamma_{\rm P}(G)=1$.}
  \label{fig:b-vertex_common}
\end{figure}

\begin{lemma}\label{lem:3-vertex_in_blue_triangle}
  If all the neighbors of a 3-vertex are b-vertices, then $G$ is isomorphic to a graph depicted in Figure~\ref{fig:b-vertex_common} or these vertices belong to a facial triakis tetrahedron as depicted in Figure~\ref{fig:triakis}.
\end{lemma}

\begin{proof}
  Let $v$ be a 3-vertex adjacent to three b-vertices $u_1$, $u_2$ and $u_3$, which necessarily form a triangle.
  By definition of a b-vertex, each $u_i$ has another 3-neighbor $v_i$.
  If the vertices $v_i$ are not all distinct, then there exist two b-vertices sharing two adjacent 3-vertices,
  and Lemma~\ref{lem:shared-3-neighbors} concludes.
  So assume the $v_i$ are distinct.
  Let $w_1$ and $w_2$ be the neighbors of $v_1$ distinct from $u_1$ (which are both adjacent to $u_1$).
  Since $u_1$ may not have any other neighbor, we infer without loss of generality that $w_2$ is adjacent to $u_2$ (and $w_1$ to $u_3$), and therefore that $w_2$ is also adjacent to $v_2$ (and $w_1$ to $v_3$).
  Similarly, $v_2$ and $v_3$ must have some vertex $w_3$ as a common neighbor, also adjacent to $u_2$ and $u_3$.
  Now, since the neighborhoods of 3-vertices and b-vertices are fully determined, we get a facial triakis tetrahedron, as depicted in Figure~\ref{fig:triakis}.
\end{proof}

\begin{figure}[ht]
  \centering
  \includegraphics[scale=0.8]{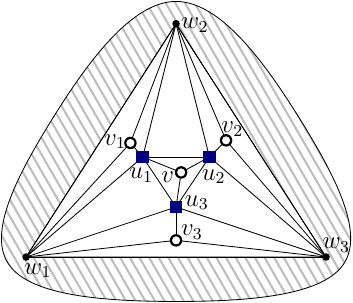}
  \caption{A facial triakis tetrahedron.}
  \label{fig:triakis}
\end{figure}

Observe in particular that if a b-vertex has three adjacent b-vertices, then by Observation~\ref{obs:adj}, we are in the case of Lemma~\ref{lem:3-vertex_in_blue_triangle} (and the graph is a triakis tetrahedron).

\begin{lemma}\label{lem:facial_blue_triangle}
  If $G$ contains three b-vertices forming a three cycle,
  then either Lemma~\ref{lem:3-vertex_in_blue_triangle} applies,
  or $G$ contains the  first configuration depicted in Figure~\ref{fig:b-vertex_C3},
  or it is isomorphic to one of the last two graphs depicted in Figure~\ref{fig:b-vertex_C3}.
\end{lemma}

\begin{proof}
  Let $u_1,u_2,u_3$ be three b-vertices forming a cycle. If they have a common 3-neighbor, then Lemma~\ref{lem:3-vertex_in_blue_triangle} applies, so assume they do not.
  By Observation~\ref{obs:adj}, every two of these vertices have a 3-vertex as a common neighbor, and they are distinct by hypothesis.
  Let $v_1,v_2,v_3$ be the 3-vertices adjacent respectively to $u_1$ and $u_2$,
  $u_2$ and $u_3$, and $u_1$ and $u_3$, and let $z_1,z_2,z_3$ be the (not necessarily distinct) third neighbors of respectively $v_1$, $v_2$ and $v_3$.
  Suppose two $z_i$ are distinct, say $z_1$ and $z_3$, and observe that the neighbors of $u_1$ are exactly $\{u_2,u_3,v_1,v_3,z_1,z_3\}$.
  Therefore, if $(u_1u_2u_3)$ separates $z_1$ from $z_3$, then $z_3$ is adjacent to $u_2$ and $z_1$ is adjacent to $u_3$.
  Now, since $u_2$ is a b-vertex, then $v_2$ is adjacent to $z_3$ and $z_1$, a contradiction.
  So $[u_1u_2u_3]$ does not separate any two $z_i$ and is facial.
  Moreover, if say $z_1$ and $z_3$ are distinct, then they must be adjacent since $u_1$ has no other neighbor.
  So depending on whether the $z_i$ are distinct or not,
  $G$ contains the first configuration depicted in Figure~\ref{fig:b-vertex_C3},
  or is isomorphic to one of the last two graphs depicted in Figure~\ref{fig:b-vertex_C3} (note that all faces incident to a 3-vertex or a b-vertex in these drawings are facial).
\end{proof}

\begin{figure}[ht]
  \centering
  \includegraphics[scale=0.8]{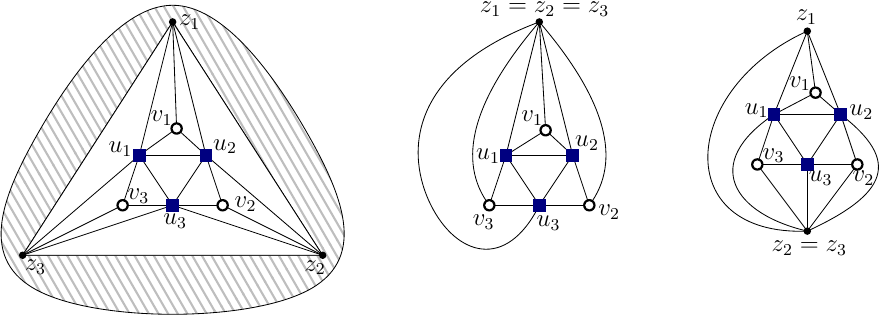}
  \caption{The possible configurations of $G$ if there is a face composed of b-vertices.
  The last two graphs, that satisfy $\gamma_{\rm P}(G)=1$, are contracts of the first configuration.}
  \label{fig:b-vertex_C3}
\end{figure}

\begin{prop} \label{prop:b-vertex}
  Let $(u_1,u_2,u_3)$ be a path on three b-vertices.
  Let $v_1$ be the 3-vertex adjacent to $u_1$ and $u_2$ and let $v_2$ be the 3-vertex adjacent to $u_2$ and $u_3$.
  If $u_1$ is not adjacent to $u_3$, then
  there exist distinct vertices $x$ and $x'$ such that $\{u_1,u_2,u_3,v_1\} \subseteq N(x)$, $\{u_1,u_2,u_3,v_2\} \subseteq N(x')$, and $[x u_2 u_3]$ and $[x' u_1 u_2]$ are facial (see Figure~\ref{fig:b-vertex_config}).
\end{prop}

\begin{proof}
  Since $v_1$ and $v_2$ are 3-vertices, then there exist two vertices $x,x'$ such that $\{u_1,u_2,v_1\}\in N(x)$ and $\{u_2,u_3,v_2\}\in N(x')$. Since $u_2$ is a b-vertex, we have that $x\neq x'$ (otherwise $u_1$ and $u_3$ would be adjacent as the second configuration of Figure~\ref{fig:b-vertex} shows). Thus, $u_2$ is a b-vertex corresponding to the first configuration of Figure~\ref{fig:b-vertex}, and so $x$ is adjacent to $u_3$, $x'$ is adjacent to $u_1$, and $[xu_2u_3]$ and $[x'u_1u_2]$ are facial.
\end{proof}

\begin{figure}[ht]
  \centering
  \includegraphics[scale=0.8]{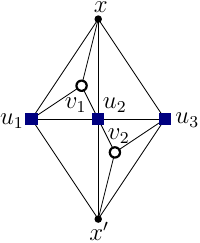}
  \caption{There are two distinct vertices both adjacent to $\{u_1,u_2,u_3\}$. All triangles are facial.}
  \label{fig:b-vertex_config}
\end{figure}

Observe that the above property together with Lemmas~\ref{lem:3-vertex_in_blue_triangle} and~\ref{lem:facial_blue_triangle} covers all possibilities of three connected b-vertices. We now consider the cases when b-vertices form paths and cycles of length at least four.

\begin{corollary}
  Suppose a set of $k \geq 3$ b-vertices form a path $(u_1,\ldots, u_k)$ (where $u_1$ and $u_k$ may be adjacent when $k>3$). Let $v_1, \ldots, v_{k-1}$ be 3-vertices with $v_i$ being adjacent to $u_i$ and $u_{i+1}$, and let $v_0$ be the 3-vertex adjacent to $u_1$ but not to $u_2$. Then there exists a vertex $x$ adjacent to all $u_i$, $1\le i \le k$ and to $v_0$ and $v_2$.
  (see Figure~\ref{fig:b-vertex_long_path}).
\end{corollary}

\begin{figure}[ht]
  \centering
  \includegraphics[scale=0.8]{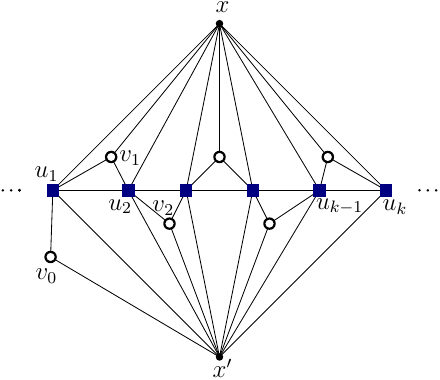}
  \caption{There are two vertices universal to the path $(u_1,u_2,u_3, \ldots, u_k)$. Vertex $x'$ is also adjacent to $v_0$ and $v_2$.}
  \label{fig:b-vertex_long_path}
\end{figure}

\begin{proof}
  Applying Proposition~\ref{prop:b-vertex} to vertices $u_1,u_2,u_3$, there exist distinct vertices $x,x'$ such that $\{u_1,u_2,$ $u_3,v_1\} \subseteq N(x)$, $\{u_1,u_2,u_3,v_2\} \subseteq N(x')$, and $[x u_2 u_3]$ and $[x' u_1 u_2]$ are facial. Since $x$ is adjacent to the b-vertex $u_3$, $x$ must be adjacent to $v_3$ and thus to $u_4$. Then $u_4$ is also adjacent to $x'$ as $u_3$ has no other neighbor.
  Iterating this argument, we infer that $x$ and $x'$ are adjacent to all $u_i$.
  Now, since $x'$ is adjacent to $u_1$ but not to $v_1$, by definition of a b-vertex it is adjacent to $v_0$, and the corollary follows.
\end{proof}

\begin{lemma}
  If $G$ contains a maximal component of b-vertices isomorphic to $P_2$, then $G$ contains a facial subgraph isomorphic to one of the graphs of Figure~\ref{fig:b-vertex_P2_subgraph}.
\end{lemma}
\begin{proof}
  Let $u_1, u_2$ be b-vertices, and let $v_1$ be the 3-vertex adjacent to $u_1$ and $u_2$, and $z$ the third neighbor of $v_1$.
  Let $v_0$ and $v_2$ be the second 3-neighbors of respectively $u_1$ and $u_2$, which can be assumed distinct by Lemma~\ref{lem:shared-3-neighbors}.
  Since $v_0$ is a 3-vertex and is not adjacent to $u_2$, $v_0$ has a neighbor $t$ which is adjacent to $u_1$ and $u_2$ (we can see $u_1$ as the central vertex of any configuration of Figure~\ref{fig:b-vertex}).
  By definition of b-vertices, $v_2$ must also be adjacent to $t$.
  Let $z_1$ and $z_2$ be the third neighbors of respectively $v_0$ and $v_2$.
  If $z_1 = z_2$, then $N[t]\subseteq N[v_0] \cup N[v_2]$, and the vertex $t$ is in fact a b-vertex, contradicting our hypothesis.
  Thus $z_1 \neq z_2$. Depending on whether $z$ and $z_1$ are distinct or not, we get one of the configurations of Figure~\ref{fig:b-vertex_P2_subgraph} (in both cases, the outer face of the drawing may not be facial).
\end{proof}

\begin{figure}[ht]
  \centering
  \includegraphics[scale=0.8]{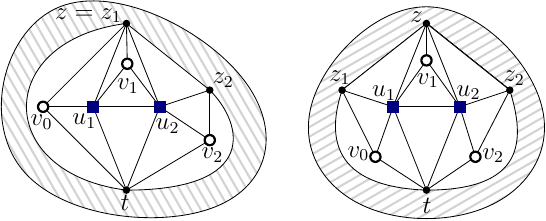}
  \caption{
    The possible configurations of $G$ if there is a $P_2$ component of b-vertices. The outer faces are not necessarily facial. All other triangles of the drawing are facial.}
  \label{fig:b-vertex_P2_subgraph}
\end{figure}

Finally, if there is an isolated b-vertex in $G$, then it belongs to one of the subgraphs depicted in Figure~\ref{fig:b-vertex}. This concludes the proof of the following lemma, that gives a characterization of the possible intersections of the configurations from Figure~\ref{fig:bad-guys}.

The \emph{special configurations} of $G$ are then all the configurations depicted in Figure~\ref{fig:final_configs}.

\begin{lemma}\label{lem:special_configs}
    If $G$ contains a special configuration as facial subgraph, then either $G$ is a small graph (characterized in Lemmas~\ref{lem:shared-3-neighbors}, \ref{lem:3-vertex_in_blue_triangle}, and \ref{lem:facial_blue_triangle}) and $\gamma_P(G) \leq \frac{n-2}{4}$, or each maximal component of b-vertices of $G$ belongs to one of the induced configurations depicted in Figure~\ref{fig:final_configs}, 1 to 7, or $G$ contains a facial octahedron (configuration 8 in Figure~\ref{fig:final_configs}).
\end{lemma}

\begin{figure}[h]
  \centering
  \includegraphics[scale=0.75]{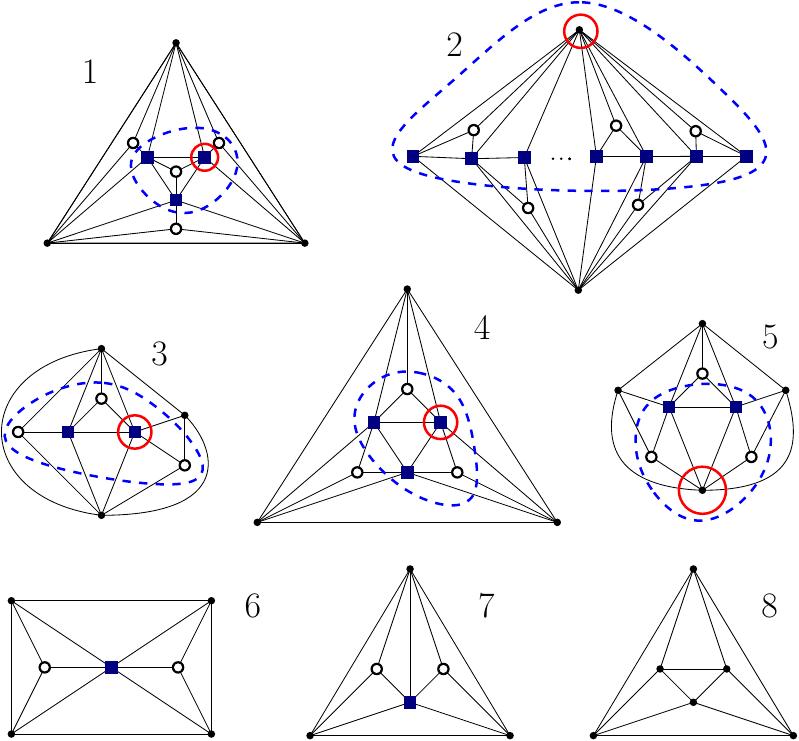}
  \caption{
    The different configurations containing a b-vertex, and the octahedron.
  So-called ``special'' vertices of configurations 1 to 5 are circled in red. For these configurations, vertices circled with a blue-dashed curve form a set relative to the special vertex of the configuration, and they are called the \emph{circled} vertices of the configuration.  
}
  \label{fig:final_configs}
\end{figure}

\begin{obs}\label{obs:disjoint}
  If a vertex belongs to two facial subgraphs isomorphic to configurations from Figure~\ref{fig:final_configs}, then it is a vertex from the outer face for both of them.
\end{obs}

\begin{proof}
  Let $v$ be a vertex that belongs to two configurations of Figure~\ref{fig:final_configs}.
  If $v$ is a b-vertex, then none of the two configurations is an octahedron.
  Then by maximality of the components of b-vertices in each configuration, the two configurations must rely on the same set of b-vertices,
  so they are the same configuration.
  Now suppose $v$ is a 3-vertex.
  In both configurations it must be an internal vertex, and have an adjacent b-vertex.
  So the configurations also share a b-vertex and the same argument concludes.
  Finally, if $v$ is a vertex of degree 4, it is an internal vertex of an octahedron.
  Since two octahedra cannot intersect on internal vertices and no internal vertex of an octahedron may be adjacent to a 3-vertex,
  $v$ does not belong to any other configuration.
  The observation follows.
\end{proof}

\section{Constructing the power dominating set}\label{sec:algo}

We now describe the process that defines incrementally a power dominating set $S$ of $G$ satisfying the announced bound.
In Section~\ref{sec:octa}, Algorithm~\ref{alg:configs} produces a set $S_1$ monitoring special configurations from Figure~\ref{fig:bad-guys} with a small number of vertices.
Then, Algorithm~\ref{alg:pds2} of Section~\ref{subsec:expansion} builds a set $S_2$ by expanding the set $S_1$ iteratively, while keeping certain properties. If the graph $G$ is not fully monitored after that, we show in Section~\ref{subsec:wrapping} that $G$ has a characterized structure, which guarantees that our last Algorithm~\ref{alg:last} maintains the wanted bound while adding some well chosen vertices to $S_2$ to build the required set $S$.

\smallskip
During the three algorithms, we ensure the following property on the set of selected vertices, that is necessary for the proof of Lemma~\ref{lem:margin}:
\smallskip

\textbf{Property $(*)$.} \emph{
  We say that a subset $S$ of vertices of a plane graph $G$ has \emph{Property $(*)$} in $G$ whenever, for each induced triangulation $G'\subseteq G$ of order at least 4, if $G'$ is monitored by $S$ then one of the following holds:     \begin{itemize}
    \item[(a)] one vertex of the outer face of $G'$ has its closed neighborhood in $G$ monitored by $S$,
    \item[(b)] or, we have $|S \cap V(G')| \leq \frac{|V(G')|-2}{4}$.
  \end{itemize}}

\subsection{Monitoring special configurations} \label{sec:octa}

The first step of our algorithm is described in Algorithm~\ref{alg:configs}, which takes care of monitoring vertices creating special configurations.
In the following, we say a configuration is \emph{monitored} by $S_1$ when all its interior b-vertices are in $M(S_1)$, or for an octahedron, if all its vertices are in $M(S_1)$.

\begin{algorithm}[h]
  \DontPrintSemicolon
    \SetAlgoLined

    \KwIn{A triangulation $G$ of order $n \geq 6$.}

    \KwOut{A set $S_1\subseteq V(G)$ monitoring all b-vertices and all vertices of facial octahedra.
    }

    $S_1:=\emptyset$\;
                \If{$G$ has a vertex $u$ of degree at least $n-2$}
                {Label $N[u]$ with $u$\;
                 Return $\{u\}$}

                \If{$G$ is a triakis tetrahedron (as in Figure~\ref{fig:triakis})}
                {$u,v$: two b-vertices of $G$ at distance 2\;
                Label $u$, its 3-neighbors and two of its adjacent b-vertices with $u$ \;
                Label all other vertices of $G$ with $v$\;
                Return $\{u,v\}$
                }

    \While {$\exists$ a non-monitored configuration $H$ from Figure~\ref{fig:final_configs}(1,2,3,4,5)}
    { $u$: the special vertex of $H$ \;
      $S_1 \leftarrow S_1\cup \{u\}$\;
        Label $u$ and the circled vertices of $H$ with $u$ \;
    }

    \While {$\exists$ non-monitored configurations $H,H'$ from Figure~\ref{fig:final_configs}(6,7,8) with a common vertex $u$
    }{
      $S_1 \leftarrow S_1\cup \{u\}$\;
        Label $u$ and the interior vertices of $H$ and $H'$ adjacent to $u$ with $u$\;
      }

    \While {$\exists$ a non-monitored configuration $H$ from Figure~\ref{fig:final_configs}(6,7,8)}{
      $u$: any exterior vertex of $H$\;
        $S_1 \leftarrow S_1\cup \{u\}$\;
        Label all vertices of $H$ with $u$\;
      }
    Return $S_1$\;
    \caption{Monitoring special configurations}
    \label{alg:configs}
\end{algorithm}

Note that the output of Algorithm~\ref{alg:configs} is the empty set whenever $G$ contains neither b-vertices nor facial octahedra.
We prove the following lemma:

\begin{lemma} \label{lem:configs}
  Let $S_1$ be the set obtained by application of Algorithm~\ref{alg:configs} to $G$. The following statements hold:
    \begin{enumerate}
      \item[(i)] All b-vertices and all facial octahedra are monitored by $S_1$.
      \item[(ii)] If $S_1$ is not empty, $|S_1| \leq \frac{|M(S_1)|-2 }{4} $.
      \item[(iii)] $S_1$ has Property $(*)$ in $G$.
    \end{enumerate}
\end{lemma}

\begin{proof}
  (i) Every b-vertex in the graph belongs to one of the configurations of Figure~\ref{fig:final_configs}. The selected vertices in each configuration monitor all the b-vertices of the configuration, and thus the algorithm monitors all such vertices. Taking any vertex of a facial octahedron monitors the whole octahedron, thus all facial octahedra are monitored as well.

  (ii) If the graph is dealt with by the first \emph{if}s, the statement is straightforward. Otherwise, we first ensure that for each vertex $u \in S_1$, there are indeed at least five vertices labeled with $u$. For configurations 1, 3, 4 and 5, this is clear by definition of the circled vertices. For configuration 2, the vertex taken plus the (at least) three b-vertices of the path plus at least one 3-vertex make (at least) five labeled vertices. For every vertex $u$ added in the second \emph{while} loop, there are at least two vertices labeled with $u$ in each of the two configurations, which together with $u$ itself makes five vertices.
  For vertices added in the last \emph{while} loop, at least six vertices are labeled with $u$ each time.

  Now, we show that each vertex receives at most one label during Algorithm~\ref{alg:configs}. By Observation~\ref{obs:disjoint}, only vertices on the outer face of some configuration may be labeled several times, and so in only two cases: they may receive their own label when they are themselves added to $S_1$, or they may receive a label during the last \emph{while} loop if they are in a non-monitored configuration 6, 7 or 8 disjoint from all remaining non-monitored configurations. Since these last configurations are monitored by any vertex of their outer face, all vertices are labeled at most once.

  If $S_1$ contains two or more vertices at the end of the algorithm, the statement is proved. If $S_1$ is reduced to a singleton, since the chosen vertex is of degree at least five, the statement holds.

  (iii) Let $G'$ be an induced triangulation of $G$ monitored after Algorithm~\ref{alg:configs}. If $|S_1 \cup V(G')| = 0$, then Property $(*)$.(b) holds. Assume then $|S_1 \cup V(G')| > 0$.
  If there is a vertex $v \in S_1$ such that some vertices labeled with $v$ are not in $G'$, then $v$ is a vertex of the outer face of $G'$ and Property $(*)$.(a) holds.
  Otherwise, for every vertex $v \in S_1 \cap V(G')$, all vertices with label $v$ (which are at least five as said above) are in $G'$. If $|S_1 \cap V(G')|\ge 2$, then this is sufficient to deduce that Property $(*)$.(b) holds. Otherwise, we observe that the set of vertices bearing a same label $u$ either does not form an induced triangulation or is of size at least six, so $G'$ contains at least six vertices and the statement also holds.
\end{proof}

In the following, $S_1$ denotes the output of Algorithm~\ref{alg:configs} applied to the graph $G$. Note that we can now forget the labels put on vertices during Algorithm~\ref{alg:configs}.

\subsection{Expansion of $S_1$}\label{subsec:expansion}

The next step consists in selecting greedily any vertex that increases the set of monitored vertices by at least four.
We first make a small observation.

In the following, the graphs of the form $P_2+P_k$ (i.e., formed by two vertices both adjacent to all vertices of a path $P_k$) for some $k\ge 1$ are called \emph{tower graphs}.
We remark that the only maximal planar graphs of order $n\leq 6$ are the complete graphs $K_3$ and $K_4$, the graphs $P_2+P_3$ and $P_2+P_4$, the octahedron, and the flip-octahedron (see Figure~\ref{fig:small_tri}).

\begin{figure}[ht]
  \centering
        \includegraphics[width=0.13\textwidth]{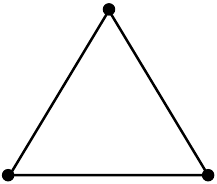} ~
        \includegraphics[width=0.13\textwidth]{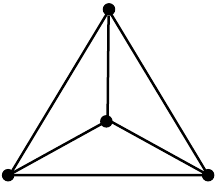} ~
        \includegraphics[width=0.13\textwidth]{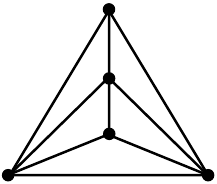} ~
        \includegraphics[width=0.13\textwidth]{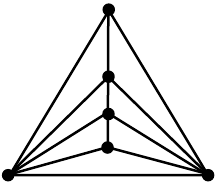} ~
        \includegraphics[width=0.3\textwidth]{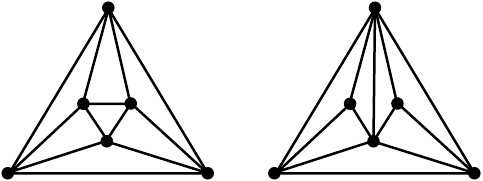}
    \caption{The maximal planar graphs of
    order $n\leq 6$.}\label{fig:small_tri}
\end{figure}

\begin{obs} \label{rem:internal}
  Let $G$ be a triangulation.
    Unless $G$ is an octahedron or a tower graph $P_2 + P_k$ (for some $k\ge1$), one interior vertex of $G$ has degree at least 5.
\end{obs}

\begin{proof}
  Suppose $G$ is not an octahedron or a tower graph. If $G$ is a flip-octahedron (last configuration of Figure~\ref{fig:small_tri}), then one of its interior vertices has degree five. Otherwise, by the preceding observation, $G$ contains at least seven vertices.
    Suppose by way of contradiction that all interior vertices of $G$ have degree at most 4.
    Denote by $u$ the exterior vertex of $G$ with maximum degree, $v,w$ the other two exterior vertices of $G$,
    and $u_1, \ldots, u_k$ the interior neighbors of $u$ ($k\ge2$ or $G$ is $K_4$), so that $(v u_1 \ldots u_k w)$ form a cycle.
    Without loss of generality, we assume that $v$ is adjacent to no less vertices among $u_1,\ldots, u_k$ than $w$ is.

    Let $\ell$ be the maximum integer such that for all $i\le\ell$, $u_i$ is adjacent to $v$. Since $G$ is not a tower graph, $\ell<k$.
    Observe that since $v$ is not adjacent to $u_{\ell+1}$,
    $u_\ell$ and $v$ have a common neighbor $t$ (that is neither $u_{\ell-1}$ nor $u$) to make another face on the edge $vu_\ell$.
    If $\ell >1$, then $v$, $t$, $u$, $u_{\ell-1}$ and $u_{\ell+1}$ make five neighbors to $u_\ell$, a contradiction.

    So $\ell=1$ and since $u_2$ is not adjacent to $v$, $t\neq u_2$. (Note that $t\neq w$ or $v$ would have only one neighbor among $u_1,\ldots,u_k$ while $w$ has at least two, contradicting our assumption.)
    Thus $u_1$ has at least four neighbors: $u$, $v$, $u_2$ and $t$. By our initial assumption, $[u_1 u_2 t]$ is a facial triangle.
    Now if $k\ge 3$, $u_2$ also already has four neighbors so $[u_2u_3t]$ form a facial triangle.
    But then $t$ and $u_3$ are already of degree four, so it is not possible to form another facial triangle containing the edge $tu_3$, a contradiction.
    So $k=2$, and $[u_2 w t]$ is facial.
    But then we get an induced octahedron where the only non facial triangle is $[v t w]$,
    in which adding a vertex would raise the degree of $t$ to more than $4$, a contradiction. This concludes the proof.
\end{proof}

Let us now proceed with the second part of the algorithm defining a power dominating set.
Assume that after Algorithm~\ref{alg:configs}, $M(S_1) \neq V(G)$.
We now apply Algorithm~\ref{alg:pds2} that builds a set of vertices $S_2 \subset V(G)$ by iteratively expanding $S_1$ in such a way that each addition of a vertex increases by at least four the number of monitored vertices. Moreover, at each round, the vertex added to $S_2$ has maximal degree in $G$ among all candidate vertices.

\begin{algorithm}
  \DontPrintSemicolon
    \SetAlgoLined

    \KwIn{A triangulation $G$ of order $n\geq 6$}

    \KwOut{A set $S_2 \subseteq V(G)$ with $|S_2| \leq \frac{|M(S_2)|-2}{4}$}

    $S_2 := \mbox{ Algorithm~\ref{alg:configs}}(G)$\;
    $M := M(S_2)$\;
    \While { $\exists$ $u$ in $V(G)\setminus S_2$ such that $|M(S_2\cup \{u\})|\geq |M|+4$}{
      Select such a vertex $u$ of maximum degree in $G$.\\
        $S_2 \leftarrow S_2\cup \{u\}$\;
    $M\leftarrow M(S_2)$\;}
  Return $S_2$\;
    \caption{Greedy selection of vertices to expand $S_1$}
    \label{alg:pds2}
\end{algorithm}
\medskip

We now prove the following lemma:

\begin{lemma} \label{lem:step3}
  Let $S_2$ be the output of Algorithm~\ref{alg:pds2} applied to $G$. The following statements hold:
    \begin{enumerate}
      \item[(i)] $|S_2| \leq \frac{|M(S_2)|-2}{4}$.
      \item[(ii)] $S_2$ has Property $(*)$ in $G$.
    \end{enumerate}
\end{lemma}

\begin{proof}

  (i)
    Let $\ell$ denote the number of rounds of Algorithm~\ref{alg:pds2} (i.e., the number of vertices added during the ``while'' loop).
    For $0 \le i\le \ell$, we denote by $S_2^{(i)}$ the set of selected vertices after the $i$-th round of Algorithm~\ref{alg:pds2} (where $S_2^{(0)}$ denotes the result of Algorithm~\ref{alg:configs} applied to $G$),
    and by $M^{(i)}$ the set $M(S_2^{(i)})$ of vertices monitored by $S_2^{(i)}$.
    The algorithm ensures that for all $0 \leq i \leq \ell-1 $, we have that if $S_2^{(i)}$ is not the empty set, then $|S_2^{(i+1)}| = |S_2^{(i)}|+1$ and $|M^{(i+1)}| \geq |M^{(i)}|+4$.
    So, provided we can establish a base case (either for $i=0$ or $i=1$), Statement (i) holds by induction on $\ell$.
    If $S_2^{(0)}$ is not the empty set,
    then $1\leq |S_2^{(0)}| \leq   \frac{|M^{(0)}|}{6}$, and thus  $|S_2^{(0)}| \leq \frac{|M^{(0)}|-2}{4}$.
    Otherwise, by Observation~\ref{rem:internal}, the first vertex added to $S_2$ is of degree at least 5 so $|M^{(1)}| \geq 6$.
    Thus this time $\frac{|M^{(1)}|-2}{4} \geq 1=|S_2^{(1)}|$, and the desired result follows by induction.

    (ii)  Let $G'\subseteq G$ be an induced triangulation monitored
    by $S_2$ after Algorithm~\ref{alg:pds2}.
    First assume $G'$ is isomorphic to a tower graph. Note that no vertex selected during Algorithm~\ref{alg:configs} is an interior vertex of a tower graph.
    Observe that for each interior vertex $v$ of a tower graph, there exists an exterior vertex $v'$ such that $N[v] \subseteq N[v']$ and $d(v') > d(v)$.
    Then at any given round $i$, Algorithm~\ref{alg:pds2} would rather select $v'$ instead of any interior vertex $v$, and thus no interior vertex of $G'$ is in $S_2$.
    Since $G'$ is monitored, then at least one of the exterior vertices of $G'$ (say $u$)
    is in $S_2$ or has propagated to an interior vertex of $G'$,
    so $N[u] \subseteq M$ and Statement (a) of Property~$(*)$ holds for $S_2$ in $G$.
    If $G'$ is isomorphic to the octahedron or to the flip-octahedron, then one of the exterior vertices of $G'$ is in $S_1$ thanks to Algorithm~\ref{alg:configs}, and Property~$(*)$.(a) also holds.

    Assume now that $|V(G')|\ge 6$, and suppose that Property~$(*)$.(a) does not hold for $G'$.
    Then some vertices of $G'$ belong to $S_2$, and vertices of $V(G')\cap S_2$ only monitor vertices of $G'$ (no propagation may occur from a vertex of the outer face of $G'$).
    Then the same proof as for (i) above restricted to $G'$ shows that Property $(*)$.(b) holds.
    This proves that Property~$(*)$ holds for $S_2$ in $G$.
\end{proof}

In the following, $S_2$ denotes the output of Algorithm~\ref{alg:pds2} on the graph $G$.

\subsection{Monitoring the remaining components}\label{subsec:wrapping}

After Algorithm~\ref{alg:pds2}, some vertices of the graph may still remain non-monitored. Algorithm~\ref{alg:last} thus completes the set $S_2$ into a power dominating set of $G$, while keeping the wanted bound. In order to succeed, we need to have a better understanding of the structure of the graph around these non-monitored vertices. More precisely, we show that the graph can be described in terms of {\em splitting structures}, (see Figure~\ref{fig:five_cases}): they are structures composed of a set $C = \{u_1,u_2,u_3\}$ of three non-monitored vertices and of two \emph{associated triangulations} $G_1$ and $G_2$ whose exterior vertices are monitored.

We make use of the following lemma, that is proved in Section~\ref{sec:proof}. We denote by $\overline{M}_G(S)$ the set of vertices not monitored by $S$ in $G$ (i.e., $V(G) \setminus M_G(S)$).

\begin{lemma}\label{lem:urien-wqt}
  Let $G$ be a triangulation, $S$ a subset of vertices of $G$ monitoring all b-vertices and facial octahedra. Let $G'$ an induced triangulation of $G$. If $\overline{M}_G(S)\cap V(G')\neq \emptyset$, and for any $v\in V(G)$, $|M_G(S\cup\{v\})|\le |M_G(S)|+ 3$ (i.e., Algorithm~\ref{alg:pds2} stopped), then $G'$ corresponds to one of the configurations depicted in Figure~\ref{fig:five_cases}.
\end{lemma}

\begin{figure}[tbp]
  \begin{center}
    \begin{tabular}{ccc}
      \includegraphics[scale=0.6,page=2]{fig_octa.pdf} &
            \includegraphics[scale=0.4]{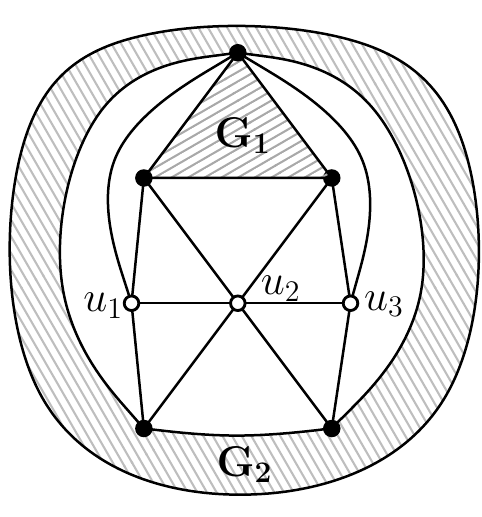} &
            \includegraphics[scale=0.4]{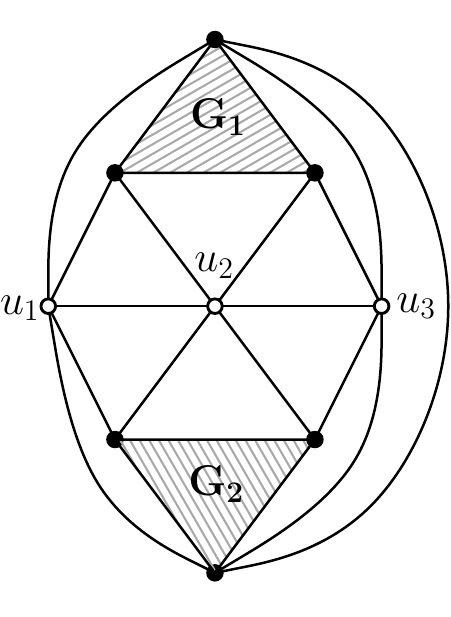} \\
            (a) & (b) & (c) \\
    \end{tabular}

        \begin{tabular}{ccc}
          \includegraphics[scale=0.5]{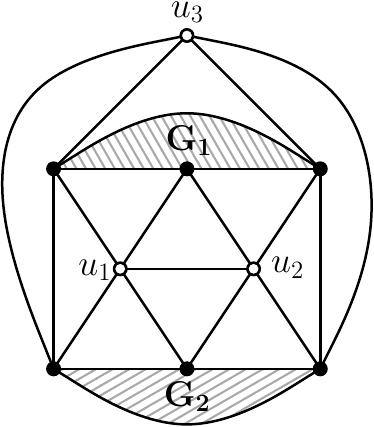} &
            \includegraphics[scale=0.5]{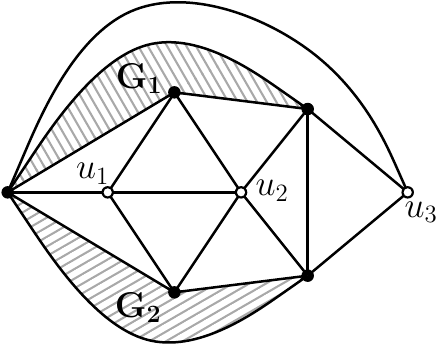} \\
            (d) & (e) \\
        \end{tabular}

        \vspace{5pt}
        \begin{tabular}{ccc}
          \includegraphics[scale=0.6]{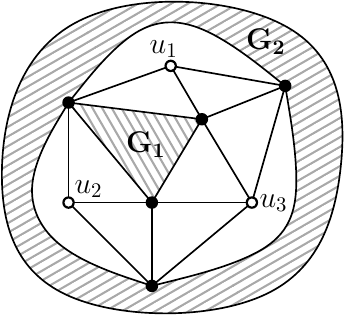} &
            \includegraphics[scale=0.55]{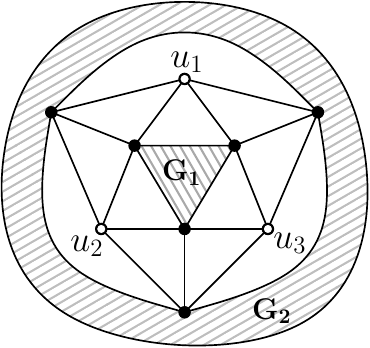} \\
            (f) & (g)\\
        \end{tabular}
  \end{center}
    \caption{The seven different splitting structures and their associated triangulations $G_1$ and $G_2$. White vertices are non monitored.
    All triangles are facial except for $G_1$ and $G_2$.}\label{fig:five_cases}
\end{figure}

Observe that the triangulations associated to a splitting structure may contain non-monitored vertices, in which case we can again apply the above lemma and deduce that they are in turn isomorphic to a splitting structure.

If $M(S_2) = V(G)$, then by Lemma~\ref{lem:step3}, $S_2$ is a power dominating set of $G$ with at most $\frac{n-2}{4}$ vertices.
Otherwise, Algorithm~\ref{alg:last}
recursively goes down to splitting structures whose associated triangulations are completely monitored, in which case it adds a vertex to $S$ to monitor the remaining vertices.

\begin{algorithm}
  \DontPrintSemicolon
    \SetAlgoLined

    \KwIn{A triangulation $G$ of order $n \geq 6$ and an induced triangulation $G'\subseteq G$}

    \KwOut{A set $S \subseteq V(G')$ monitoring $G'$ and such that $|S|\leq \frac{|V(G')|-2}{4}$}
    $S \leftarrow V(G') \cap $Algorithm~\ref{alg:pds2}$(G)$\;
    {
      \If {$\exists$ $u \not\in M_G(S)$}{
        $G_1,G_2 \leftarrow$ triangulations associated to the splitting structure of $G'$ containing $u$\;
            $S'$ $\leftarrow$ Algorithm \ref{alg:last}($G$, $G_1$) \;
            $S''$ $\leftarrow$ Algorithm \ref{alg:last}($G$, $G_2$) \;
        $S \leftarrow S' \cup S'' \cup \{u\}$\;}
    }
    Return $S$\;
    \caption{Monitoring the last vertices}
    \label{alg:last}
\end{algorithm}

We now prove that after the addition of vertices during Algorithm~\ref{alg:last}, the wanted bound still holds.

\begin{lemma} \label{lem:margin}
  Let $G'$ be an induced triangulation of $G$, and $S$ a subset of vertices of $G$ monitoring all b-vertices and facial octahedra. Let $C$ be a splitting structure in $G'$ with $G_1$ and $G_2$ its associated triangulations. Let $u$ be a vertex of $C$, and let $S'$ denote the set $S \cap V(G_1)$ and $S''$ the set $S \cap V(G_2)$. If $G_1$ and $G_2$ are monitored by $S$ and $S'$ and $S''$ have Property $(*)$ respectively in $G_1$ and $G_2$, then $S' \cup S'' \cup \{u\}$ has Property $(*)$ in $G'$, and $G'$ is monitored.
\end{lemma}

\begin{proof}
  First recall that after application of Algorithm~\ref{alg:pds2}, any vertex in $M_G(S)$ has at most three non-monitored neighbors. Therefore, in the induced triangulation $G'$, a vertex adjacent to a vertex in $C$ may not be adjacent to vertices from another configuration $C'$ in $G$, or it would have two non-monitored neighbors in $C$ and two in $C'$, a contradiction. Thus if a vertex can propagate in $G'$, then it can also propagate in $G$.

    We know that $S'$ and $S''$ have Property $(*)$ in repectively $G_1$ and $G_2$, and so $G_1$ and $G_2$ both satisfy either Property $(*)$.(a) or (b).
    Since all exterior vertices of $G_1$ and $G_2$ have non-monitored neighbors, then in fact, $G_1$ and $G_2$ satisfy Property $(*)$.(b). Thus $|S'| \leq \frac{|V(G_1)|-2}{4}$ and $|S''| \leq \frac{|V(G_2)|-2}{4}$.
    We remark that $|V(G')| \geq |V(G_1)| + |V(G_2)| + 2$ in every splitting structure. After adding a vertex $u \in C$, we have:
    \[|S' \cup S'' \cup \{u\}| \leq \frac{|V(G_1)|-2}{4} + \frac{|V(G_2)|-2}{4} + 1 = \frac{|V(G_1| + |V(G_2)|}{4} \leq \frac{|V(G')|-2}{4} \,.\]
    Moreover, the exterior vertices of induced triangulations of $G'$ all have only monitored neighbors (the exterior vertices of $G'$ excepted) and thus $S' \cup S'' \cup \{u\}$ has Property $(*)$ in $G'$.

        To prove that the addition of one vertex of $C$ is sufficient to monitor $G'$, we consider different cases depending on the splitting structure.
    \begin{itemize}
      \item For splitting structures (a), (b) and (c), adding $u_2$, then $u_1$ and $u_3$ are monitored by adjacency.
      \item For splitting structures (d) and (e), adding $u_1$, then $u_2$ is monitored by adjacency, and then any vertex of the outer face of $G_1$ or $G_2$ propagates to $u_3$.
      \item For splitting structures (f) and (g), adding $u_1$, then two exterior vertices of $G_1$ propagate independently to the other two vertices of $C$.
    \end{itemize}
    Thus $G'$ is monitored, which concludes the proof.
\end{proof}

We can now use Lemma~\ref{lem:margin} to prove by direct induction on the splitting structures that at the end of Algorithm~\ref{alg:last}, Property $(*)$ holds for $S$ in $G$. Moreover, the proof of Lemma~\ref{lem:margin} shows that for the set $S$, $G$ satisfies Property $(*)$.(b). Thus the output $S$ of Algorithm~\ref{alg:last} satisfies the wanted bound and the graph is completely monitored.
We thus get the following corollary that concludes the proof of Theorem~\ref{th:main}.

\begin{corollary} \label{cor:end}
  At the end of Algorithm~\ref{alg:last}, $M(S) = V(G)$ and $|S| \leq \frac{|V(G)|-2}{4}$.
\end{corollary}

In the following section, we finally prove Lemma~\ref{lem:urien-wqt}.

\section{Defining splitting structures} \label{sec:proof}

This section is dedicated to the proof of Lemma~\ref{lem:urien-wqt}.
In the following, we work under the assumption of the lemma,
i.e., we assume that the set $S$ monitors all octahedra and b-vertices,
and that the addition of any vertex $v$ to $S$ would extend the set of vertices monitored by $S$ by at most three.
Any vertex contradicting the second part of the assumption is called a \emph{contradicting vertex}.
For simplicity, when $G$ and $S$ are clear from context, we denote $M = M_G(S)$ and $\overline{M} = V(G) \setminus M_G(S)$.

As a direct consequence of the definition of power domination, we get the following observation:

\begin{obs}
  Let $S$ be a set of vertices of $G$ such that for every vertex $v \in V(G)$, $|M_G(S \cup \{v\})| \leq  |M_G(S)| + 3$. The following properties hold:
  \begin{enumerate}
    \item[(i)]\label{obs:degreeMnew} Each vertex of $M$ has either zero, two or three non-monitored neighbors.
    \item[(ii)]\label{obs:degreeoM} Each vertex of $\overline{M}$ has at most 2 neighbors in $\overline{M}$.
    \item[(iii)]\label{lem:propagation}
      For every vertex $u \in M \setminus S$, there exists $v \in M \cap N(u)$ such that $N[v] \subset M$ (that propagated to $u$).
  \end{enumerate}
\end{obs}

\medskip
We now make the following statement.

\begin{lemma}
  \label{lem:chainsaw2} If $v$ is of degree at least five, then for every two neighbors $u_1$ and $u_2$ of $v$,  there exists a neighbor $w$ of $v$ adjacent to $u_1$ or $u_2$, but not both, and the corresponding triangle $[vu_iw]$ is facial.
\end{lemma}

\begin{proof}
  We partition the set of neighbors of $v$ into two paths from $u_1$ to $u_2$: a path $(w'_1,\ldots, w'_k)$ of length at least three (i.e., $k\ge 2$) and another path $(w_1, \ldots, w_\ell)$, possibly empty. We have $w'_1\neq w'_k$. By way of contradiction, assume both $w'_1$ and $w'_k$ are adjacent to both $u_1$ and $u_2$. Contracting the path $(w'_1,\ldots, w'_{k-1})$ into $w'_1$ and the path $(u_1, w_1,\ldots, w_\ell)$ into $u_1$, we get that $u_1, u_2, v, w'_1, w'_k$ induce a $K_5$ in the resulting graph, contradicting planarity of $G$. Thus $w'_1$  is not adjacent to $u_2$ (and $[w'_1u_1v]$ is facial) or $w'_k$ is not adjacent to $u_1$ (and $[w'_ku_2v]$ is facial).
\end{proof}

We remark that Lemma~\ref{lem:chainsaw2} also holds when $v$ is in $M$ with at least two neighbors in $\overline{M}$. Indeed, by Observation~\ref{lem:propagation}, $v$ has a neighbor $v'$ that propagated to it. Then $v'$ only has monitored neighbors, and two of them are also adjacent to $v$. Thus $v$ has degree at least five, which is the hypothesis of Lemma~\ref{lem:chainsaw2}.

\begin{lemma} \label{lem:conn_comp}
  Components of $G[\overline{M}]$ are of order at most three.
\end{lemma}

\begin{proof}

  Let $C$ be a component of $G[\overline{M}]$.
    By Observation~\ref{obs:degreeoM}, each vertex of $\overline{M}$ has degree at most two in $\overline{M}$, so $C$ is a path or a cycle. Then adding any vertex of $C$ to $S$ would monitor all of $C$. Since we work under the assumption of Lemma~\ref{lem:urien-wqt}, $C$ is of order at most three.
\end{proof}

 Thus each component of $G[\overline{M}]$ is isomorphic to either $K_3$, $P_3$, $P_2$, or $K_1$. Lemmas~\ref{lem:K3}, \ref{lem:P3} and \ref{lem:P2} deal successively with the first three cases, whereas Lemma~\ref{lem:P1} goes through the case where $\overline{M}$ is an independent set in the induced triangulation considered.

\begin{lemma} \label{lem:K3}
  Let $G$ and $M$ satisfy the assumption of Lemma~\ref{lem:urien-wqt}. If an induced triangulation $G'$ contains a component of $G[\overline{M}]$ isomorphic to $K_3$, then $G'$ is isomorphic to the configuration depicted in Figure~\ref{fig:config_octa}.
\end{lemma}

\begin{figure}
  \centering
    \includegraphics[scale=0.7]{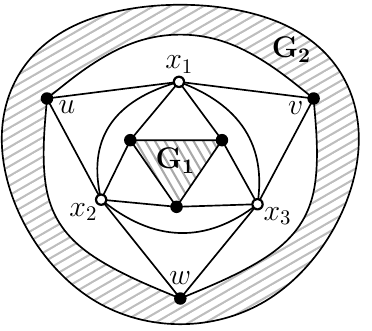}
    \caption{The configuration for $G'$ containing a non-monitored component isomorphic to $K_3$. $G_1$ and $G_2$ are triangulations. All other triangles of the drawing are facial.}
    \label{fig:config_octa}
\end{figure}

\begin{proof}
  Let $C$ be a component of $G[\overline{M}]$ isomorphic to $K_3$ with $V(C)=\{x_1,x_2,x_3\}$.
    Let $u$ be a vertex of $M$, adjacent to at least one of the vertices of $C$.

    We first consider the case when $u$ has neighbors in $\overline{M} \setminus C$. If $u$ is adjacent to two vertices in $C$, then by Observation~\ref{obs:degreeMnew}, $u$ has exactly one neighbor in $\overline{M}\setminus C$, say $v$. Then $M(S\cup\{u\})\supseteq M(S)\cup\{x_1,x_2,x_3,v\}$, and $u$ is a contradicting vertex.
    So $u$ has only one neighbor in $C$, say $x_1$.
    Within the neighborhood of $x_1$, the path from $x_2$ to $x_3$ going through $u$ must contain at least three interior vertices since $u$ is not adjacent to $x_2$ or $x_3$, so $x_1$ is of degree at least five.
    Applying Lemma~\ref{lem:chainsaw2} on $x_1$, we get that a neighbor $w$ of $x_1$ is adjacent to $x_2$ or $x_3$ but not both.
    Since $w$ is adjacent to two vertices in $C$, it has no other neighbors in $\overline{M}$ or the above case would apply.
    Hence, adding $u$ to $S$, all neighbors of $u$ in $\overline{M}$ get monitored, then $w$ propagates to $x_2$ or $x_3$ which can in turn propagate to the last vertex of $C$.
          So $u$ is a contradicting vertex.

          We assume now that any vertex of $M$ adjacent to $C$ has only vertices of $C$ as neighbors in $\overline{M}$.
          Note that such a vertex must be adjacent to at least two vertices in $C$.
    Let $u$ be a common neighbor of $x_1$ and $x_2$ such that $[u x_1 x_2]$ is facial ($u$ exists since the edge $x_1 x_2$ is contained in exactly two facial triangles).
    By Lemma~\ref{lem:chainsaw2}, there is a neighbor $v$ of $u$ that is adjacent to only one of $\{x_1,x_2\}$ (say $x_1$) and $[uvx_1]$ is facial.
    The vertex $v$ must have a second non-monitored neighbor, that must be in $C$, so $v$ is adjacent to $x_3$.
    Observe that the triangle $[vx_1x_3]$ must be facial. Otherwise, there is a vertex $t\neq v$ such that $[tx_1x_3]$ is facial and $t$ is separated from $x_2$ by $(vx_1x_3)$.
                Then by Lemma~\ref{lem:chainsaw2}, $t$ has a neighbor $t'$ with only one neighbor among $\{x_1,x_3\}$ also separated from $x_2$ by $(vx_1x_3)$, and thus with only one non-monitored neighbor, a contradiction.
    Now, $v$ and $x_3$ have a common neighbor $w$ outside the triangle $[v x_1 x_3]$, such that $[vwx_3]$ is facial.
    By definition of $v$, we have $w \neq x_2$.
                We also have $w\neq u$ or $v$ would be of degree three contradicting Observation~\ref{lem:propagation}.
    The cycle $(uvx_3x_2)$ separates $w$ from $x_1$, so the second non-monitored neighbor of $w$ (different from $x_3$) must be $x_2$.
    Unless an additional edge $uw$ form a facial triangle $[uwx_2]$, there is another neighbor of $x_2$ that is separated from both $x_1$ and $x_3$ by the cycle $(uvwx_2)$, a contradiction. So $u$ is adjacent to $w$ in a facial triangle $[uwx_2]$.

                In a similar way that we proved that $[vx_1x_3]$ is facial, we infer that $[wx_2x_3]$ is facial.
                By construction, $[ux_1x_2]$, $[uvx_1]$, $[vwx_3]$ are facial, and we proved $[vx_1x_3]$, $[uwx_2]$ and $[wx_2x_3]$ also are.
                If the triangle $[x_1x_2x_3]$ is facial, then the graph induced by the vertices $u,v,w,x_1,x_2,x_3$ is a facial octahedron, contradicting the assumption of Lemma~\ref{lem:urien-wqt}. Thus $[x_1 x_2 x_3]$ is not facial, and applying the same line of reasoning as above inside $[x_1x_2x_3]$ shows that $G'$ is isomorphic to the configuration depicted in Figure~\ref{fig:config_octa}.
\end{proof}

\begin{lemma} \label{lem:P3} Let $G$ and $M$  satisfy the assumption of Lemma~\ref{lem:urien-wqt}. If an induced triangulation $G'$ contains a component of $G[\overline M]$ isomorphic to $P_3$, then $G'$ is isomorphic to one of the splitting structures depicted in Figure~\ref{fig:P3}.
\end{lemma}

\begin{figure}
  \centering
    \includegraphics[scale=0.45,page=2]{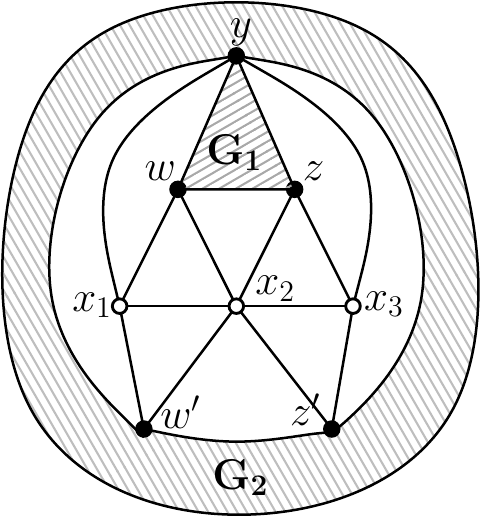}~~~~
    \includegraphics[scale=0.45,page=2]{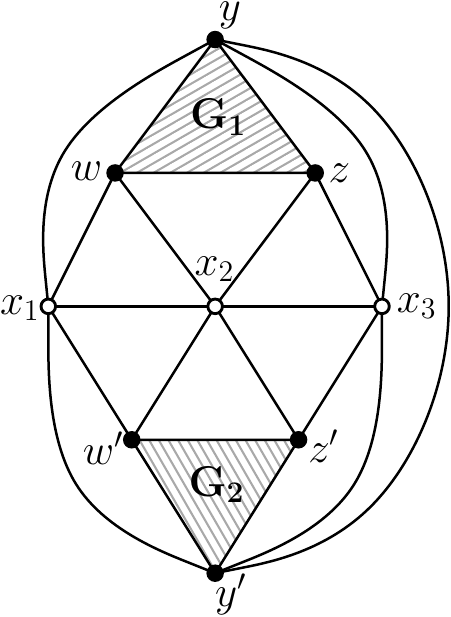}
    \caption{The two possible configurations for $G'$ containing a non-monitored component isomorphic to $P_3$. $G_1$ and $G_2$ are triangulations. All other triangles of the drawing are facial.}
    \label{fig:P3}
\end{figure}

\begin{proof}
  Let $C$ be a component of $G[\overline{M}]$ isomorphic to $P_3$ with $V(C)=\{x_1,x_2,x_3\}$.

        \medskip
        Let $u$ be a vertex adjacent to $C$.	
        We first prove that all neighbors of $u$ in $\overline{M}$ are vertices of $C$. By Observation \ref{obs:degreeMnew}, $u$ has at most two neighbors in $\overline M\setminus V(C)$.
    If $u$ has exactly one neighbor $u_1$ in $\overline M\setminus V(C)$, then $x_2$ is a contradicting vertex, since $u$ propagates to $u_1$ once $x_2$ is added to $S$.
    Assume then that $u$ has two neighbors in $\overline{M} \setminus V(C)$ and thus only one neighbor in $C$.
    If $u$ is adjacent to $x_1$ or $x_3$, then $u$ is a contradicting vertex. Suppose that $u$ is adjacent to $x_2$ only, which must then be of degree at least five. We apply Lemma~\ref{lem:chainsaw2} on $x_2$ and get a neighbor $v$ of $x_2$ adjacent to $x_1$ or $x_3$ but not both. Taking $u$ in $S$, $v$ then propagates (and then $x_2$ propagates to $x_3$) so $u$ is a contradicting vertex.
    Thus neighbors of $C$ may not be adjacent to vertices in $\overline{M}$ that are not in $C$.

    \smallskip
    We now prove that there is no vertex of $M$ adjacent to all vertices of $C$.
    Suppose by way of contradiction that $u$ is a vertex in $M$ adjacent to $x_1$, $x_2$ and $x_3$.
    By Lemma~\ref{lem:chainsaw2}, $u$ has a neighbor $z \in M$ with exactly one neighbor in $\{x_1, x_3\}$ (say $x_1$) and $[uzx_1]$ is facial. Note that by the above statement, $z$ is also adjacent to $x_2$.

    Again, we can apply Lemma~\ref{lem:chainsaw2} to find a neighbor $z'$ of $z$ adjacent to $x_1$ or $x_2$ but not both. Vertex $z'$ must have a second non-monitored neighbor, namely $x_3$. So $z'$ cannot be adjacent to $x_1$ which is separated from $x_3$ by $(ux_2z)$, so $z'$ is adjacent to $x_2$ and $x_3$ and $[x_2zz']$ is facial. Now $u$ is necessarily adjacent to $z'$ forming a facial triangle $[ux_3z']$ (otherwise some vertex would have a single neighbor in $C$).

    Observe that the triangle $[zx_1x_2]$ must be facial. Otherwise, there is a vertex $t\neq z$ such that $[tx_1x_2]$ is facial and $t$ is separated from $x_3$ by $(zx_1x_2)$.
                Then by Lemma~\ref{lem:chainsaw2}, $t$ has a neighbor $t'$ with only one neighbor among $\{x_1,x_2\}$ also separated from $x_3$ by $(zx_1x_2)$, and thus with only one non-monitored neighbor, a contradiction.
                With a similar argument, we get that $[z'x_2x_3]$, $[ux_1x_2]$ and  $[ux_2x_3]$ are facial. But then $x_2$ is a b-vertex (as in the bad configuration of Figure~\ref{fig:bad-guys}), a contradiction.

    \smallskip
    Let $w,w',z,z' \in M$ such that $[x_1x_2w],[x_1x_2w'],[x_2x_3z],[x_2x_3z']$ are faces. By the above statement, all these vertices are distinct.
    Suppose that there is a neighbor $u$ of $x_2$ different from the above vertices. Vertex $u$ has a second neighbor in $C$, say $x_1$. The cycle $(ux_1x_2)$ separates $w$ or $w'$ from $x_3$, say $w$.
    By Lemma~\ref{lem:chainsaw2}, $w$ has a neighbor with exactly one neighbor in $\{x_1,x_2\}$, and that cannot be adjacent to $x_3$, a contradiction.
    Thus $x_2$ has no other neighbor. Renaming vertices if necessary, we suppose $[x_2wz]$ and $[x_2w'z']$ are facial triangles.

    Note that $x_1$ or $x_3$ must have another neighbor. Otherwise,  $w$ is adjacent to $w'$ and$z$ is adjacent to $z'$, which implies that $x_2$ is a b-vertex (as in the ugly configuration of Figure~\ref{fig:bad-guys}), a contradiction.
     Let $y \in M$ be a neighbor of $x_1$ such that $[x_1 w y]$ is facial. The second neighbor of $y$ in $C$ is necessarily $x_3$.
    Similarly, $z$ has a neighbor $z_1$ such that $[x_3 z_1 z]$ is facial and adjacent to $x_1$ and $x_3$. Note that $z_1\neq w$ or $w$ would be adjacent to three vertices in $C$. Then $y=z_1$ or the cycle $(x_3ywz)$ would separate $z_1$ from $x_1$. We prove with similar arguments that there is a vertex $y'$ such that $[x_1 w' y']$ and $[x_3 z' y']$ are facial.

    If $y=y'$, then $G'$ is isomorphic to the first splitting structure of Figure~\ref{fig:P3}. Otherwise, suppose first that $x_1$ has another neighbor $t$. It also has to be adjacent to $x_3$. Then applying Lemma~\ref{lem:chainsaw2} to $t$, we find a vertex adjacent to only one vertex in $C$, a contradiction. So $y$ and $y'$ are adjacent, and $[x_1yy']$ and $[x_3yy']$ are facial. Thus $G'$ is isomorphic to the second configuration of Figure~\ref{fig:P3}.
\end{proof}

\begin{lemma} \label{lem:P2}
  Let $G$ and $M$ satisfy the assumption of Lemma~\ref{lem:urien-wqt}. If an induced triangulation $G'$ contains a non-monitored component isomorphic to $P_2$, then $G'$ is isomorphic to one of the configurations depicted in Figure~\ref{fig:P2}.
\end{lemma}

\begin{figure}
  \centering
    \includegraphics[scale=0.55]{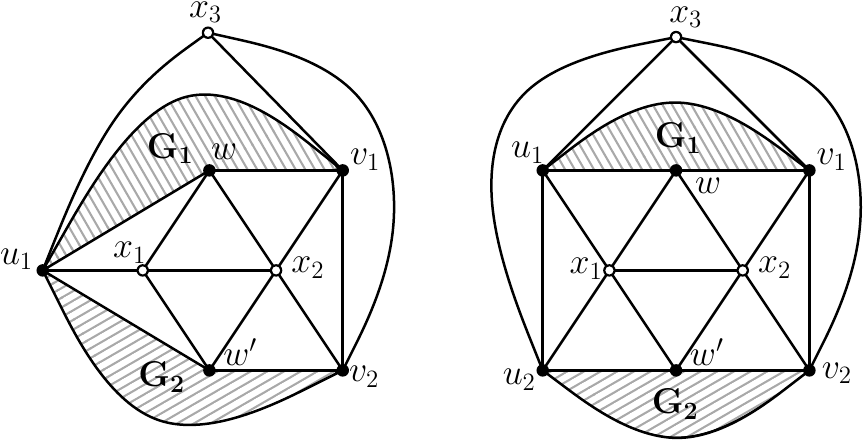}
    \caption{The two possible configurations of an induced triangulation $G'$ if $G[\overline M]$ has a component isomorphic to $P_2$. $G_1$ and $G_2$ are triangulations. All other triangles of the graph are facial.}
    \label{fig:P2}
\end{figure}

\begin{proof}

  Let $C = \{x_1,x_2\}$ with $x_1x_2 \in E(G)$, and let $w$ and $w'$ be the vertices such that $[x_1x_2w]$ and $[x_1x_2w']$ are facial.

    \medskip
    {\bf Claim 1.} There is exactly one vertex of $\overline{M}$ at distance 2 of $C$.

    \medskip
    {\it Proof. } Suppose there is no vertex of $\overline{M}$ at distance 2 from $C$. By Lemma~\ref{lem:chainsaw2}, $w$ has a neighbor $t\in M$ adjacent to only one vertex among $\{x_1, x_2\}$. Then $t$ has only one neighbor in $\overline{M}$, which contradicts Observation~\ref{obs:degreeMnew}. Thus there is a vertex of $\overline{M}$ at distance 2 from $C$.
    Suppose that there is $u' \in \overline{M} \setminus V(C)$ neighbor of a vertex $u\in N(C)$ and $v' \in \overline{M} \setminus V(C)$ neighbor of another vertex $v\in N(C)$. Then $u$ is a contradicting vertex (whether it is distinct from $v$ or not).
    {\footnotesize ($\square$) }

    \medskip
    Let $x_3$ be the only vertex of $\overline{M}$ at distance 2 of $C$.

    \medskip
    {\bf Claim 2.} The vertices adjacent to $x_3$ are exactly the vertices of $(N(x_1)\cup N(x_2)) \setminus\{w,w'\}$.

    \medskip
    {\it Proof. } Suppose there is a vertex $w'' \in M, w'' \neq \{w,w'\}$ adjacent to $x_1$ and $x_2$.
                The cycle$(x_1x_2w'')$ separates $x_3$ from either $w$ or $w'$, say $w$.
    By Lemma~\ref{lem:chainsaw2}, there exists a vertex $v$ adjacent to $w$ and to only one vertex among $\{x_1,x_2\}$.
    By Observation~\ref{obs:degreeMnew}, $v$ has a second non-monitored neighbor, that cannot be $x_3$, which contradicts Claim~1. Thus $w$ and $w'$ are the only common neighbors of $x_1$ and $x_2$.
   Therefore, all vertices adjacent to only one of $x_1$ and $x_2$ (i.e., in $N(x_1)\cup N(x_2)\setminus\{w,w'\}$) are adjacent to $x_3$ (and there is at least one such vertex).

    Suppose there exists some vertex $v$ adjacent to $x_3$ but not in $N(x_1)\cup N(x_2)$. Then $v$ is in $\overline{M}$ or it has another neighbor $x_4\in \overline{M} \setminus\{x_1,x_2,x_3\}$, and $v$ is a contradicting vertex. Thus no vertex $v\in V(G)\setminus (N(x_1)\cup N(x_2))$ is adjacent to $x_3$.

  We now prove that $w$ and $w'$ are not adjacent to $x_3$.
    Suppose $w$ is adjacent to $x_3$.
    By Lemma~\ref{lem:chainsaw2}, $w$ has a neighbor $u_1$ adjacent to only one of $\{x_1,x_2\}$ (say $x_1$) such that $[u_1x_1w]$ is facial. (Thus $u_1$ is also adjacent to $x_3$ and $[wu_1x_3]$ is facial, since it separates $x_3$ from $x_1$ and $x_2$.)
     Again by Lemma~\ref{lem:chainsaw2}, $u_1$ has a neighbor $v_1$ in $M$ adjacent to only one of $\{x_1,x_3\}$.
     Suppose first $v_1$ is adjacent to $x_3$ (and not to $x_1$). Then $v_1$ is also adjacent to $x_2$.
     Following Observation~\ref{lem:propagation}, $w$ has other neighbors in $M$ different from $u_1$. So there is a vertex $t$ such that $[x_2tw]$ is facial, and since $t$ is separated from $x_1$ by $(x_2v_1x_3w)$, $t$ is adjacent to $x_3$. Applying Lemma~\ref{lem:chainsaw2} on $t$, we get a contradiction.
     So $v_1$ is adjacent to $x_1$ but not to $x_3$, and thus $v_1 = w'$ (and $w'$ is not adjacent to $x_3$). But $x_3$ has degree at least three, so there is a vertex $v_2$ adjacent to $x_2$ and $x_3$.
     Again, $[u_1x_3w]$, $[v_2x_2w]$ and $[v_2x_3w]$ must be facial. But then there is no vertex that may have propagated to $w$.
     Thus $w$ and $w'$ are not adjacent to $x_3$.
       {\footnotesize ($\square$) }

    \bigskip
    Let us now consider the neighbors of $x_1$ and $x_2$ in $M \setminus \{w,w'\}$. Let $(u_1, \ldots, u_k)$ and $(v_1, \ldots, v_\ell)$ be the paths from $w$ to $w'$ among respectively $N(x_1) \cap M$ and $N(x_2)\cap M$. Since $x_3$ has degree at least 3, then by Claim 2, $k+\ell \ge 3$.
  First observe that $k$ and $\ell$ both are at most 2. Otherwise, say $k\ge 3$, then by Claim 2, each $u_i$ is adjacent to $x_3$, and the triangles $[u_i u_{i+1} x_3]$ are facial, in particular $[u_1 u_2 x_3]$ and $[u_2 u_3 x_3]$. But then $u_2$ contradicts Observation~\ref{lem:propagation}.

  We thus have two cases:\begin{itemize}

    \item $k+\ell=3$, say $u_1$ is the only neighbor of $x_1$ and $v_1, v_2$ are the only two neighbors of $x_2$ in $M \setminus \{w,w'\}$.
      By Claim 2, $u_1$, $v_1$ and $v_2$ are neighbors of $x_3$. Moreover, since none of $\{w,w'\}$ is adjacent to $x_3$, $u_1$ is adjacent to $v_1$ and $v_2$. Also by Claim 2, triangles $[u_1v_1x_3]$, $[v_1v_2x_3]$ and $[u_1v_2x_3]$ are facial, and $G$ is isomorphic to the first graph depicted in Figure~\ref{fig:P2}.

    \item $x_1$ and $x_2$ both have exactly two neighbors in $M \setminus \{w,w'\}$.
      By Claim 2, $u_1$, $u_2$, $v_1$ and $v_2$ are neighbors of $x_3$. Again, $u_1$ is adjacent to $v_1$ and $u_2$ is adjacent to $v_2$ since neither $w$ nor $w'$ is adjacent to $x_3$.
        Also by Claim 2, triangles $[u_1v_1x_3]$, $[v_1v_2x_3]$, $[u_1u_2x_3]$ and $[u_2v_2x_3]$ are facial and $G$ is isomorphic to the second graph depicted in Figure~\ref{fig:P2}.
  \end{itemize}
This concludes the proof. \end{proof}

\begin{lemma} \label{lem:P1}
  Let $G$ and $M$ satisfy the assumptions of Lemma~\ref{lem:urien-wqt}. If an induced triangulation $G'$ is such that $\overline{M} \cap V(G')$ is an independent set, then $G'$ is isomorphic to one of the splitting structures depicted in Figure~\ref{fig:K1_5}.
\end{lemma}

\begin{figure}
  \centering
    \includegraphics[scale=0.7]{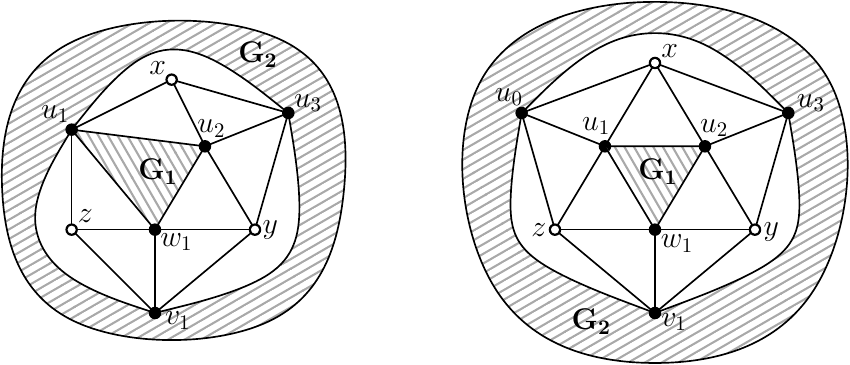}
    \caption{The two possible configurations of an induced triangulation $G'$ if $\overline M \cap V(G')$ is an independent set. $G_1$ and $G_2$ are triangulations. All other triangles of the drawings are facial.}
    \label{fig:K1_5}
\end{figure}

\begin{proof}
  Let $G$ satisfy the assumptions of the lemma. In this proof, we denote $M'$ the set $M \cap V(G')$, and $\overline{M}'$ the set $\overline{M} \cap V(G')$.

  \medskip
  {\bf Claim 1.}
  There exists a vertex $u \in M'$ with only two neighbors in $\overline{M}'$.

  \smallskip
    {\it Proof.}
        Suppose by way of contradiction that every vertex in $M'$ with non-monitored neighbors has exactly three neighbors in $\overline{M}'$. Note first that no two vertices in $M'$ have exactly two common neighbors in $\overline{M}'$, or they would be contradicting vertices. Hence, they share either one or three such neighbors.

  Suppose first that all vertices in $M'$ with a common neighbor in $\overline{M}'$ have exactly one such common neighbor.
  We define an auxiliary graph $H$ as follows: the vertices of $H$ are the vertices in $\overline{M}'$, and two vertices in $H$ are adjacent if they have a common neighbor in $G'$. Observe that from a planar drawing of $G'$, we can easily build a planar drawing of $H$: we keep the position of the vertices, and for each edge $(uv)$ in $H$, $u$ and $v$ have a common neighbor $x$ in $G'$ and we can have the edge $(uv)$ follow closely the edges $(ux)$ and $(xv)$ (that would not create crossings since $N_{G'}(x)\cap \overline{M}' =3$).
  By our assumption that no two vertices in $M'$ have more than one common neighbor in $\overline{M}'$, the degree of a vertex in $H$ is  precisely twice its degree in $G'$. Since every vertex in $\overline{M}'$ has degree at least 3 and every vertex in $M'$ has three neighbors in $\overline{M}'$, that implies that $H$ has minimum degree at least 6. But this contradicts Euler's formula for planar graphs.

  So there are at least two vertices $u$ and $v$ in $M'$ with three common neighbors in $\overline{M}'$, say $x_1$, $x_2$ and $x_3$, forming a subgraph isomorphic to a $K_{2,3}$. 
  Consider such five vertices, such that the subgraph $G''$ induced by the vertices within the outer face of the $K_{2,3}$ does not contain the same structure. Denote $x_1$ and $x_3$ the exterior vertices (i.e., $x_2$ is inside the cycle $(x_1ux_3v)$). Since $x_1$, $x_2$ and $x_3$ are pairwise non adjacent, there is another neighbor $w$ of $x_2$ in $M'$, which has at least two other neighbors in $\overline{M}'$. By minimality of the selected $K_{2,3}$, now all vertices in $G''$ that belong to $M'$ and share a neighbor in $\overline{M}'$ do share exactly one. Building the graph $H$ on $G''$ the same way as above, we get a planar graph $H$ where every vertex has degree at least six except for $x_1$, $x_2$ and $x_3$ that have respectively degree at least 2, 4 and 2. Therefore, we get that the sum of the degrees of the vertices in $H$ is at least $6|V(H)|-10$, again a contradiction with Euler's formula. This concludes the proof.
    {\footnotesize ($\square$) }

  \medskip
  {\bf Claim 2.}
  If a vertex $u$ of $M'$ has degree 2 in $\overline{M}'$, then all the vertices of $M'$ sharing a neighbor in $\overline{M}'$ with $u$ also have degree 2 in $\overline{M}'$.

        \smallskip
        {\it Proof.}
        Let $u_1$ be a vertex of $M'$ with two neighbors $v_1, v_2$ in $\overline{M}'$. Suppose that there exists a vertex $u_2$ in $M'$ adjacent to $v_1$ or $v_2$ (say $v_1$) and with degree 3 in $\overline{M}'$. If $u_2$ is not adjacent to $v_2$, then $u_2$ is a contradicting vertex.
        So assume $u_2$ is also adjacent to $v_2$ and let $v_3$ be the third neighbor of $u_2$ in $\overline{M}'$.
        Applying Lemma~\ref{lem:chainsaw2} to vertex $u_1$, let $t$ be a vertex adjacent to only one of $v_1$ and $v_2$, say $v_1$, and such that $[u_1 v_1 t]$ is facial. There is another vertex adjacent to $t$ in $\overline{M}'$. If this vertex is not $v_3$, then $t$ is a contradicting vertex (as $u_1$ propagates to $v_2$ then $u_2$ to $v_3$). So $t$ has only two neighbors in $\overline{M}'$, $v_1$ and $v_3$.

        Now, every other vertex in the graph is separated from $v_1$, $v_2$ or $v_3$ by one of the three separating cycles $(u_1v_1u_2v_2)$, $(tv_1u_2v_3)$ and $(tu_1v_2u_2v_3)$. The monitored vertex $u_2$ necessarily has more neighbors (by Observation~\ref{lem:propagation}). Suppose there is a neighbor of $u_2$ in the cycle $(tu_1v_2u_2v_3)$. Then there is a neighbor $w$ to $u_2$ and $v_2$ forming a face $[u_2v_2w]$. If $w$ is not adjacent to $v_3$, then $w$ has some extra neighbors in $\overline{M}'$, and is a contradicting vertex ($u_1$ propagates to $v_1$ then $u_2$ to $v_3$). If $w$ is also adjacent to $v_3$, by Lemma~\ref{lem:chainsaw2} it has a neighbor adjacent to only one of $v_2$ and $v_3$, also separated from $v_1$ by the cycle $(tu_1v_2u_2v_3)$, and the same argument applies. The same arguments apply also if $u_2$ has neighbors in the other separating cycles.	
        Thus there is no vertex adjacent to $v_1$ or $v_2$ with degree 3 in $\overline{M}'$.
  This concludes the proof.
        {\footnotesize ($\square$) }

    \medskip

                Let $u_1$ be a vertex of $M'$ with exactly two neighbors in $\overline{M}'$, denoted $x$ and $z$.
                By Lemma~\ref{lem:chainsaw2}, there is a neighbor of $u_1$ adjacent to only one of $x$ and $z$,
                say $u_2$ is adjacent to $x$ but not $z$ (and $[xu_1u_2]$ is facial).
                By Claim~2, $u_2$ has only one other neighbor in $\overline{M}'$, denote it $y$.
         Note that we now have the property
    {\bf (P)}: any neighbor $v \in M'$ of $x$, $y$ or $z$ has at least two neighbors in $\{x,y,z\}$ and is not adjacent to any vertex of $\overline{M}' \setminus \{x,y,z\}$. Otherwise $v$ would be a contradicting vertex.
                Consider the two paths from $u_2$ to $z$ that partition $N(u_1)$.
                Let $w_1$ be the last vertex before $z$ in the path that does not go through $x$ (i.e., $[u_1w_1z]$ is facial and $w_1$ is not adjacent to $x$). Since $u_2$ is not adjacent to $z$, then $w_1 \neq u_2$. By the above property (P), $w_1$ is adjacent to $y$. Moreover, $y$ may not have a neighbor separated from $x$ and $z$ by $(u_1u_2yw_1)$, so $[u_2w_1y]$ is a facial triangle.

          Suppose first that $x$ is of degree three, and let $u_3$ denote its third neighbor, adjacent to both $u_1$ and $u_2$. It has one other neighbor among $y$ and $z$, say $y$. Observe that $[u_2u_3y]$ is necessarily a facial triangle, and that by Claim~2, $u_3$ is not adjacent to $z$. Since $z$ is of degree at least $3$, it has a neighbor $v_1 \neq u_3$ such that $[u_1v_1z]$ is facial. By property (P), $v_1$ is adjacent to $y$. Now $z$ has no other neighbor within the cycle $(v_1yw_1z)$, or it would be a common neighbor to $y$ and $z$, but applying Lemma~\ref{lem:chainsaw2} would lead to a contradiction. So $w_1$ is adjacent to $v_1$, and  $[v_1w_1y]$ and $[v_1w_1z]$ are facial triangles. In addition, $y$ cannot have a neighbor separated from $x$ and $z$ by $(u_1u_3yv_1)$ so $[u_3 v_1y]$ is a facial triangle. Thus we are in the first configuration of Figure~\ref{fig:K1_5}.

  Assume now that each of $x$, $y$ and $z$ have degree at least 4.
  Let $u_3$ form a facial triangle with $x$ and $u_2$. If $u_3$ is adjacent to $z$, then the fourth neighbor of $x$ is also adjacent to $z$. By Lemma~\ref{lem:chainsaw2}, it has a neighbor adjacent to only one of $x$ and $z$, which is separated from $y$ by $(u_1xu_3z)$, a contradiction. So $u_3$ is adjacent to $y$ forming a facial triangle $[u_2u_3x]$. By the same argument, we infer the existence of $u_0$ and $v_1$, common neighbors of $x$ and $z$ and of $y$ and $z$ respectively, and that the corresponding triangles are facial. If $x$ were of degree $5$, then we would get similarly a contradiction applying Lemma~\ref{lem:chainsaw2} on $u_3$ or $u_0$. By the same reasoning on $y$ and $z$, we obtain the second configuration of Figure~\ref{fig:K1_5}, which concludes the proof of Lemma~\ref{lem:P1}.
\end{proof}

The results from the four Lemmas~\ref{lem:K3}, \ref{lem:P3}, \ref{lem:P2} and \ref{lem:P1} conclude the section: after Algorithm~\ref{alg:pds2}, each induced triangulation of $G$ is isomorphic to one of the graphs depicted in Figure~\ref{fig:five_cases}.

\nocite{*}
\bibliographystyle{abbrvnat}
\bibliography{biblio_power_dom}
\label{sec:biblio}

\end{document}